\date{Rev. 10/IV/13 JM }
\newif\ifcmts
\cmtsfalse
\title{Three-monotone interpolation\thanks{This 
research was started at the 3rd KAM{\'A}K workshop held in Vranov
nad Dyj{\'i}, Czech Republic, September 15-20, 2013,
which was supported by the grant SVV-2013-267313 (Discrete
Models and Algorithms).
J.C. was also supported by this grant. J.M. was supported
by the  ERC Advanced Grant No.~267165.
P.P. was  supported by the grant SVV-2014-260107
 }
}

\author{
{\sc Josef Cibulka}\\
   {\footnotesize Department of Applied Mathematics}\\[-1.5mm]
   {\footnotesize  Charles University, Malostransk\'{e} n\'{a}m. 25}\\[-1.5mm]
{\footnotesize  118~00~~Praha~1, Czech Republic, and}\\
   {\footnotesize Institute of Physics of the ASCR, v.v.i. }\\[-1.5mm]
   {\footnotesize  Za Slovankou 1782/3, 182~00~~Praha~8}\\[-1.5mm]
{\footnotesize  Czech Republic}\\[-1.5mm]
{\footnotesize e-mail: {\tt cibulka@kam.mff.cuni.cz}}
\and
{\sc Ji\v{r}\'{\i} Matou\v{s}ek}
\\
   {\footnotesize Department of Applied Mathematics}\\[-1.5mm]
   {\footnotesize  Charles University, Malostransk\'{e} n\'{a}m. 25}\\[-1.5mm]
{\footnotesize  118~00~~Praha~1,
   Czech Republic, and}\\
{\footnotesize    Department of  Computer Science}\\[-1.5mm]
{\footnotesize    ETH Zurich,
      8092 Zurich, Switzerland}
\\[-1.5mm]   {\footnotesize e-mail: {\tt matousek@kam.mff.cuni.cz}}
\and
{\sc Pavel Pat\'{a}k}\\
   {\footnotesize Department of Algebra}\\[-1.5mm]
   {\footnotesize  Charles University, Sokolovsk\'{a} 83}\\[-1.5mm]
{\footnotesize  186~75~~Praha~8, Czech Republic}\\[-1.5mm]
{\footnotesize e-mail: {\tt ppatak@seznam.cz}}
}

\documentclass[11pt]{article}

\usepackage{a4}
\usepackage{latexsym}
\usepackage{amssymb,amsmath,amsthm}
\usepackage{graphicx}
\usepackage{url}

\newtheorem{theorem}{Theorem}[section]

\newtheorem{observation}[theorem]{Observation}
\newtheorem{lemma}[theorem]{Lemma}
\newtheorem{question}[theorem]{Question}

\newtheorem{claim}[theorem]{Claim}
\newtheorem{problem}[theorem]{Problem}
\newtheorem{corol}[theorem]{Corollary}

\newcommand{\heading}[1]{\vspace{1ex}\par\noindent{\bf\boldmath#1}}

\makeatletter
\newcommand{\ProofEndBox}{{\ifhmode\unskip\nobreak\hfil\penalty50 \else
          \leavevmode\fi\quad\vadjust{}\nobreak\hfill$\Box$
            \finalhyphendemerits=0 \par}}
\makeatother

\newcommand{\R}{{\mathbb{R}}}

\newcommand{\Q}{{\mathbb{Q}}}

\newcommand\eps{\varepsilon}

\newcommand{\sgn}{\mathop {\rm sgn}\nolimits}

\numberwithin{equation}{section}

\newcommand\MM{\mathbf M}

\newcommand\makevec[1]{{\bf #1}}

\def \ttt {\makevec{t}}

\def\:{\colon}

\long\def\onefigure#1#2{
\begin{figure*}[tbp]
\begin{center}
#1
\end{center}
\caption{#2}
\end{figure*}
}

\def\immediateFigure#1{%
\smallskip\begin{center}#1\end{center}\smallskip }

\newcommand{\labfig}[2]  
{\onefigure{\mbox{\includegraphics{#1}}}{\label{f:#1} #2} }

\newcommand{\labfigw}[3]  
{\onefigure{\mbox{\includegraphics[width=#2]{#1}}}{\label{f:#1} #3}}

\newcommand{\immfig}[1]  
{\immediateFigure{\mbox{\includegraphics{#1}}}}

\newcommand{\immfigw}[2] 
{\immediateFigure{\mbox{\includegraphics[width=#2]{#1}}}}

\ifcmts
\newcommand{\marrow}{\marginpar{\boldmath$\longleftarrow$}}
\newcommand{\jirka}[1]{\ifhmode\newline\fi\marrow \textsf{*** (JIRKA: ) #1\newline}}
\newcommand{\pepa}[1]{\ifhmode\newline\fi\marrow \textsf{*** (PEPA: ) #1\newline}}

\else
\newcommand{\marrow}{}
\newcommand{\jirka}[1]{}

\fi

\begin{document}

\maketitle

\begin{abstract} A function $f\:\R\to\R$ is called \emph{$k$-monotone}
if it is $(k-2)$-times differentiable and its $(k-2)$nd derivative is convex.
A point set $P\subset\R^2$  is \emph{$k$-monotone interpolable}
if it lies on a graph of a $k$-monotone function. These notions
have been studied in analysis, approximation theory etc. since the 1940s.

We show that 3-monotone interpolability
is very non-local: we exhibit an arbitrarily large finite $P$ 
for which every proper subset is $3$-monotone 
interpolable but $P$ itself is not.
On the other hand, we prove a Ramsey-type result:
 for every $n$ there exists $N$ 
such that every $N$-point $P$ with distinct $x$-coordinates
contains an $n$-point $Q$ such that $Q$ or its vertical
mirror reflection are $3$-monotone interpolable. The analogs
 for $k$-monotone interpolability with $k=1$ and $k=2$
are classical theorems of Erd\H{o}s and Szekeres, while
the cases with $k\ge4$ remain open.

We also investigate the computational complexity of 
deciding $3$-mono\-tone interpolability
of a given point set. Using a known characterization, this decision problem
can be stated as an instance of polynomial optimization
and reformulated as a semidefinite program. We exhibit an example
for which this semidefinite program has only doubly exponentially large
feasible solutions, and thus known algorithms cannot solve it in polynomial
time. While such phenomena have been well known for semidefinite programming
in general, ours seems to be the first such example in polynomial
optimization, and it involves  only univariate quadratic
polynomials.

\end{abstract}

\section{Introduction}

\heading{Generalizing two theorems of Erd\H{o}s and Szekeres. }
This research was inspired by two famous 1935 theorems 
of Erd\H{o}s and Szekeres \cite{es-cpg-35}. The first one
asserts that for every $n$ there is $N$ such that every 
sequence $P$ of $N$ points in the plane with increasing
 $x$-coordinates contains
an $n$-point nonincreasing or nondecreasing subsequence
(see, e.g., Steele \cite{steele-surv} for 
six nice proofs and some applications), 
and the second theorem
makes an analogous statement about the existence of
an $n$-point convex or concave subsequence (see, e.g., Morris and Soltan
 \cite{MorrisSoltan} for proofs and a survey
of developments around this result). 

For our purposes, a \emph{nondecreasing sequence} can
be defined as one lying on the graph of a nondecreasing
function $\R\to\R$, and similarly for nonincreasing, convex, and
concave sequences.
Eli\'a\v{s} and Matou\v{s}ek \cite{highes-aim}
suggested a generalization where one looks
for a subsequence lying on the graph of a function whose $k$th
derivative is nonnegative or nonpositive.
Here we consider this question but in a slightly different and technically
more convenient formulation.
(Let us remark that a number of other generalizations 
of the Erd\H{o}s--Szekeres theorems
have recently been considered 
\cite{FoxPachSudSuk,cfpss-semialg,BukhMat,Suk-OT,EMRS}.)

\heading{$k$-monotone functions. } The following five-point set
\immfig{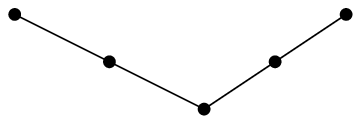}
lies on the graph of a convex function but not on the graph of
a convex twice differentiable function. This illustrates that
the requirement as above, with a function whose $k$th derivative
is nonnegative or nonpositive, is not technically quite suitable.
\jirka{It is still possible that we could Ramsey out a subset
lying on the graph of a function with $f^{(k)}$ having a constant sign.?!?}

In \cite{highes-aim} this kind of issues was circumvented by assuming
sufficiently general position of $P$. However, there is a well-established
notion of $k$-monotonicity of a function, which seems perfectly suitable 
for our purposes and does not require any general position
assumption. 

Namely, for $k\ge 2$, a function $f$ is
\emph{$k$-monotone} on an open interval $I$ if  
its $(k-2)$nd derivative $f^{(k-2)}$ 
(exists and) is convex on $I$.
(With some fantasy, this definition can also be applied for $k=1$
and leads to the usual notion of a nondecreasing function.)

Note that $k$-monotonicity is of the ``nondecreasing'' kind,
while the corresponding ``nonincreasing'' notion has $f^{(k-2)}$
concave. The term ``$k$-monotone'' may thus be somewhat confusing
in this respect, since ``monotone function'' usually means
nondecreasing \emph{or} nonincreasing,
but it seems well established in the literature.

The notion of $k$-monotonicity goes back to 
Schoenberg's 1941 abstract \cite{Schoe41}, preceded by a still older
notion of a \emph{completely monotone function}. It has been
studied from various angles in a number of papers in relation
to integral representations of functions, approximation theory,
probability, etc. We refer to Williamson \cite{WilliamsonMultiply}
for an early study\footnote{Let us remark that some of the literature,
especially older one such as \cite{Schoe41,WilliamsonMultiply}, 
the definition of $k$-monotonicity
is somewhat different, also involving requirements on lower-order 
derivatives, but the essence of the notion remains the same. 
The term \emph{$k$-convex} is also used instead of $k$-monotone.}
and to Pe\v{c}ari\'c et al.~\cite{Pecaric-al}
and Roberts and Varberg~\cite{RobertsVarberg} for
various properties and applications; for our investigations
we mostly rely on Kopotun and Shadrin~\cite{kopotun-shadrin}.

\heading{A Ramsey-type result for $3$-monotone interpolability. }
Let us call a set $P\subset\R^2$ \emph{$k$-monotone interpolable}
if it lies on a graph of a $k$-monotone function. The 
question about generalizing the Erd\H{o}s--Szekeres theorems
to $k$-monotonicity can be stated as follows:

\begin{question}\label{q:}
For which $k\ge 3$ does the following hold?
For every integer $n$ there exists $N=N_k(n)$ such that
every $N$-point $P\subset\R^2$ with distinct $x$-coordinates
contains an $n$-point subset $Q$ such that $Q$ or $Q^\updownarrow$ is
$k$-monotone interpolable (where $Q^\updownarrow$ 
denotes the mirror reflection of $Q$ about the $x$-axis).
\end{question}

In Section~\ref{s:rams3} we provide a positive answer for
$k=3$.

\begin{theorem}\label{t:rams3} The statement in Question~\ref{q:}
holds for $k=3$.
\end{theorem}

Unfortunately, our proof does not seem to generalize
to any larger $k$, and so Question~\ref{q:} remains open for $k\ge 4$.

\heading{A nonlocal behavior of $3$-monotone interpolability. }
An obvious necessary condition for a set $Q$ to be $k$-monotone
interpolable is that every $(k+1)$-tuple in $Q$ be $k$-monotone
interpolable, and for $k\le 2$ it is easy to check that this
is also sufficient.

In earlier versions of \cite{highes-aim}, it was conjectured
that the condition should be sufficient for all $k\ge 3$.
If this were the case, then Theorem~\ref{t:rams3} would follow
immediately from Ramsey's theorem for fourtuples.

However, Rote found a counterexample for $k=3$
(reproduced in \cite{highes-aim}): a six-point
set $P$ for which all fourtuples are  $3$-monotone interpolable,
but $P$ itself is not.
Later we learned that a similar example was known earlier
 \cite[Example~5.3]{kopotun-shadrin}.

In Section~\ref{s:nonhelly}
we provide a much stronger example showing that $3$-monotone interpolability
is a completely global property.

\begin{theorem}\label{t:nohelly} For every even $n \ge 4$ there exists
an $n$-point $P\subset\R^2$ that is not $3$-monotone interpolable,
but for which every proper subset is $3$-monotone interpolable.
\end{theorem}

This, in our opinion, makes Theorem~\ref{t:rams3} somewhat
surprising and Question~\ref{q:} for $k\ge 4$ interesting.

It is straightforward to extend our proof of Theorem~\ref{t:nohelly}
to yield an analogous result for every \emph{odd} $k\ge 3$. The case of
even $k$ seems somewhat more problematic, although we believe that
the difficulties should not be unsurmountable.

\heading{The algorithmic question. } We also investigate the computational
complexity of the question, Given a finite $P$ in the plane,
is it $k$-monotone interpolable? 

This is a numerical problem, and so it is important to specify the
model of computation, and also to distinguish exact and approximate
version of the question. 

We will use the \emph{bit model}
(or \emph{Turing machine model}) of computation, where one counts
the number of bit operations; thus, for example, the
addition of two $b$-bit numbers takes time proportional to~$b$. 
We assume that the  coordinates of the points of the input set 
$P$ are rational numbers, and
the size of $P$ is measured as the number of bits 
in its binary encoding (each of the rational coordinates is
encoded by the numerator and denominator written in binary).
See, e.g., Gr\"otschel, Lov\'asz, and Schrijver
\cite{gls-gaco-88} for more details on this model of
computation. 

Let us remark that for geometric computations, 
the \emph{real RAM}, or \emph{Blum--Shub--Smale}, 
model is also used in many papers, 
where arithmetic operations with arbitrary real numbers 
are allowed at unit cost. 
However, for testing $k$-monotone interpolability,
we believe that this model is inadequate,
since as we will show, a natural algorithm for this testing
needs to deal with numbers having exponentially many digits.

Kopotun and Shadrin \cite{kopotun-shadrin} provided a characterization
of $k$-monotone interpolability, which we will recall 
in Section~\ref{s:ks-char} below. Using this characterization
and methods of polynomial 
optimization, as discussed e.g. in Lasserre's book \cite{lasserre-book}, 
one can write down a semidefinite program
that is feasible if and only if the given point set $P$
is \emph{not} $3$-interpolable. (We will provide a brief discussion
of semidefinite programming and basic references in Section~\ref{s:sdp}.)

In our experience, many people in theoretical computer science
regard semidefinite programs more or less
automatically as polynomial-time solvable problems. 
(Some of the authors certainly did belong among
these people before working on the present paper.)
Indeed, many introductory texts and classes may make this impression,
although they usually point out that the known polynomial-time
algorithms solve semidefinite programs only approximately.

However, for the polynomiality claim to be true, one also needs
to assume that, if  the semidefinite program in question is feasible
at all, it has a feasible solution with norm bounded by an integer $R$
with polynomially many bits (polynomially in the size of the input).
It is known that such a bound need not hold in general
and that the smallest feasible solution may need exponentially many bits,
but in many applications of semidefinite programming, e.g., in 
combinatorial optimization, it is obvious that such a pathology
cannot occur. 

In contrast, for the semidefinite program mentioned above corresponding to 
$3$-monotone interpolability, we found that there are simple input
point sets that do force the smallest feasible solution to have
exponentially many digits. This result, Corollary~\ref{c:expdigiSDP} below,
is based on the following example.

\begin{theorem} \label{t:expdigits}
Let $P_i=\{z,p_0, p_1, \ldots, p_{2m+1},q\}$,
where $z = (-1,0)$, $p_j=(j,j^3)$ for $j=0,1,\ldots,2m+1$,
and $q=(2m+2,(2m+2)^3-6)$. Let $P'_m=(P_m\setminus\{q\})\cup\{q'\}$,
where $q'$ is $q$ shifted upwards by $2\cdot 2^{-2^{m}}$.
Then $P'_m$ is $3$-monotone interpolable, while $P_m$ is not.
\end{theorem}

The best known algorithm for deciding feasibility of an arbitrary
semidefinite program we could find in the literature is
due to Porkolab and Khachiyan \cite{Porkolab97onthe}, and it has
 exponential complexity (more precisely, the time complexity
is at most $\exp(O(s\log s))$, where $s$ is the input size). 
This also yields the best complexity
of an exact algorithm for testing $k$-monotone interpolability we are 
aware of (another algorithm of comparable complexity can be obtained 
from algorithms for deciding sentences in the first-order theory
of the reals, which are discussed, e.g., in book Basu, Pollack, and Roy 
\cite{BasuPollackRoy-book}, but here we will not consider
this alternative approach).

\heading{Future work. } We consider the Ramsey-theoretic question,
about the existence of a large $k$-monotone interpolable subset in
any sufficiently large point set, interesting and unusual in the
context of geometric Ramsey theory, because of the nonlocal nature
of $k$-monotone interpolability. The open case $k\ge 4$ seems to
need a new idea. Another question is estimating the order of magnitude
of the Ramsey function $N_3(n)$.

On the computational side, the problem of (exact) testing $k$-monotone
interpolability can be regarded as a simple concrete
instance of polynomial optimization in the spirit of \cite{lasserre-book}.
Thus, it would be very interesting to obtain stronger hardness results,
or possibly an algorithm with provably subexponential complexity.

For semidefinite programming,
there is a lower bound result of Tarasov and Vyalyi \cite{tarasovVyalyi}:
the problem of deciding feasibility of a semidefinite
program (exactly) is at least as hard as the following problem:
given an integer arithmetic circuit without inputs, determine the sign
of its output. This is a problem of basic importance for many 
complexity questions of numerical mathematics (see,
e.g., Allender et al.~\cite{allender2009complexity}), and its
complexity status is unknown and probably  very challenging to determine.
Can an analog of the Tarasov--Vyalyi 
result be obtained for some simple case of polynomial
optimization, such as the non-positivity problem (stated 
later as Problem~\ref{p:pos})?
Or perhaps even for the very specific case of
testing $3$-monotone interpolability? 

According to Ramana \cite{ramana97},
given a semidefinite program $\Pi$, one can construct another semidefinite
program, the \emph{Ramana dual} of $\Pi$, that is feasible
iff $\Pi$ is infeasible, and whose input size is
bounded by a polynomial in the input size of $\Pi$.
Thus, testing feasibility of a semidefinite program
is, in this sense, symmetric with respect to the YES and NO answers;
for example, it either belongs to both NP and co-NP, or it is outside
of both NP and co-NP. Can a similar result be obtained for polynomial
optimization, and/or for $3$-monotone interpolability?

Our example in Theorem~\ref{t:expdigits} indicates that at least
the ``obvious'' certificates of $3$-monotone \emph{noninterpolability}
are not of polynomial size. Is there a polynomial-size certificate
for $3$-monotone \emph{interpolability}, or some result indicating that
such a certificate is unlikely to be found?

One might also seek an ``elementary''
algorithm for deciding $3$-monotone interpolability, 
say one trying to combine an interpolant from
suitable parabolic arcs.
\jirka{I did not have the energy to write about
the ``naive parabolic algorithm'' and its failure.
We might still make a remark on that.}

Finally, in spite of our negative examples, one may hope that
the $k$-monotone interpolability problem, at least
for not too many points, is ``usually'' solvable in practice
by running a semidefinite solver on the semidefinite program
set up in Section~\ref{s:sdp}. For this to have at least some
theoretical foundation, it would be good to have an approximation
result of the following kind: \emph{There is an
algorithm that, given $k$, a point set $P$, and 
a parameter $\eps>0$, returns YES or NO, and runs in time polynomial
in $k$, the input size of $P$, and $\log\frac1\eps$. If the answer is
NO, then $P$ is not $k$-monotone interpolable, and if
the answer is YES, then there is a $k$-monotone interpolable
set $P'$ that can be obtained from $P$ by shifting 
every point up or down by at most~$\eps$.}

Currently we do not have such a result. There are theoretical bounds,
based on the ellipsoid method, on the complexity of approximately
solving semidefinite
programs in the bit model; see, e.g., \cite[Thm.~2.6.1]{GaertMat}
for a concrete formulation based on general theorems of \cite{gls-gaco-88}.
However, the main difficulty one faces when trying to apply
such a bound to polynomial optimization is that the ellipsoid algorithm,
in order to be guaranteed to find a feasible solution, needs
that the set of feasible solutions be suitably bounded (which can be arranged
in our setting) and contains an $\eps$-ball, for $\eps>0$ with polynomially
many bits. (The ball is not in the space of \emph{all}
positive semidefinite matrices, but rather in the space of all
such matrices satisfying all equality constraints of the
semidefinite program.) The latter condition, for semidefinite
programs coming from polynomial optimization problems, looks
at least non-obvious, and perhaps it might even fail
in some cases. 

We believe that this kind of theory is worth working out,
preferably in the general context of multivariate polynomial
optimization as in \cite{lasserre-book}---at least we could not
find any study in this direction.

\section{Preliminaries}\label{s:ks-char}

\heading{Divided differences and $k$-monotonicity. }
The \emph{$k$th divided difference} of a real function $f$
at points $x_0,x_1,\ldots,x_{k}\in \R$ is denoted by
$[x_0,x_1,\ldots,x_k]f$ and defined recursively by
\[
[x_0]f:=f(x_0), \ \ [x_0,\ldots,x_k]f:=
\frac{[x_1,\ldots,x_k]f-[x_0,\ldots,x_{k-1}]f}{x_k-x_0}.
\]
It is known that $f$ is $k$-monotone on an open interval $I$
iff $[x_0,x_1,\ldots,x_k]f\ge 0$ for all
choices of $x_0<x_1<\cdots<x_k\in I$ 
(see \cite[Lemma~3.1]{kopotun-shadrin}).

Sometimes it will be notationally convenient to regard
a set $P$ of points in the plane with distinct $x$-coordinates
as the graph of a function $f\:X\to \R$, where
$X=X(P)$ is the set of the $x$-coordinates of the points of $P$.
Then, instead of $P$ being $k$-monotone interpolable,
we can also say that $(X,f)$ is $k$-monotone interpolable.

Here is a useful criterion for determining the sign
of the divided difference $[x_0,x_1,\ldots,x_k]f$,
where $x_0<x_1<\cdots<x_k$:
Let $i\in \{0,1,\ldots,k\}$, and let $p$ be the unique polynomial
of degree at most $k-1$ such that $p(x_j)=f(x_j)$ for
all $j\in \{0,1,\ldots,k\}\setminus \{i\}$. Then 
$\sgn [x_0,x_1,\ldots,x_k]f= (-1)^{k-i}\sgn(p(x_i)-f(x_i))$
(see \cite{highes-aim}). So, for example, for $k=3$,
if we pass a parabola through the first three values
of $f$, then the fourth value is above the parabola
for $[x_0,x_1,x_2,x_3]f>0$, and below it for
$[x_0,x_1,x_2,x_3]f<0$.

A necessary condition for $k$-monotone interpolability of
$(X,f)$ is $[x_0,\ldots,x_k]f\ge 0$ for every
choice of $x_0<x_1<\cdots<x_k\in X$. While, as was discussed
in the introduction, this condition is very far from sufficient
for arbitrary $X$, it is sufficient for $|X|=k+1$ (e.g., because
$[x_0,\ldots,x_k]f$ is the leading coefficient of the
unique polynomial $p$ of degree at most $k$ that coincides with
$f$ on $X$, and if this coefficient is nonnegative, then
$p$ is a $k$-monotone interpolant; see, e.g.,  \cite{highes-aim}).


\heading{A representation theorem for $k$-monotone functions. }
The following characterization of $k$-monotone function
essentially goes back to Schoenberg \cite{Schoe41}; see
\cite{kopotun-shadrin}.

\begin{theorem}[Representation theorem]\label{t:repr}
A function $f\:\R\to\R$ is $k$-monotone if and only if for
every closed interval $[a,b]$ there is a polynomial
$p(x)$ of degree at most $k-1$ and a bounded nondecreasing
function $\mu\:[a,b]\to \R$ such that
\[
f(x)=p(x)+\frac1{k!}\int_a^b k\max(x-t,0)^{k-1}{\rm d}\mu(t), \ \ x\in[a,b].
\]
\end{theorem}

This basically says that a $k$-monotone function must be a 
nonnegative linear combination of translates of the function $\max(x,0)^{k-1}$,
plus a polynomial of degree at most $k-1$ (except that we do not have
a finite linear combination but an integral). In particular, a $3$-monotone
function can be made of a parabola and ``right half-parabolas''.

\heading{A characterization of \boldmath$k$-monotone interpolability. }
Let $X=\{x_1,\ldots,x_{n+k}\}\subset \R$,
$x_1<x_2<\cdots<x_{n+k}$, be a set of $n+k$
real numbers, which are often referred to as \emph{nodes} in this context.
The \emph{B-splines} of degree $k-1$ corresponding to $X$ are
the functions $M_1(t)$,\ldots, $M_n(t)$ defined by the formula
\[
M_i(t):= k[x_i,\ldots,x_{i+k}]\max(0,x-t)^{k-1},
\]
where the divided differencing on the right-hand side
is with respect to $x$ (while $t$ is viewed as a fixed parameter).
Here is an example with $k=3$ (the nodes 
are marked on the $x$-axis, and the peaks of $M_1,\ldots,M_5$
go in the left-to-right order):
\immfigw{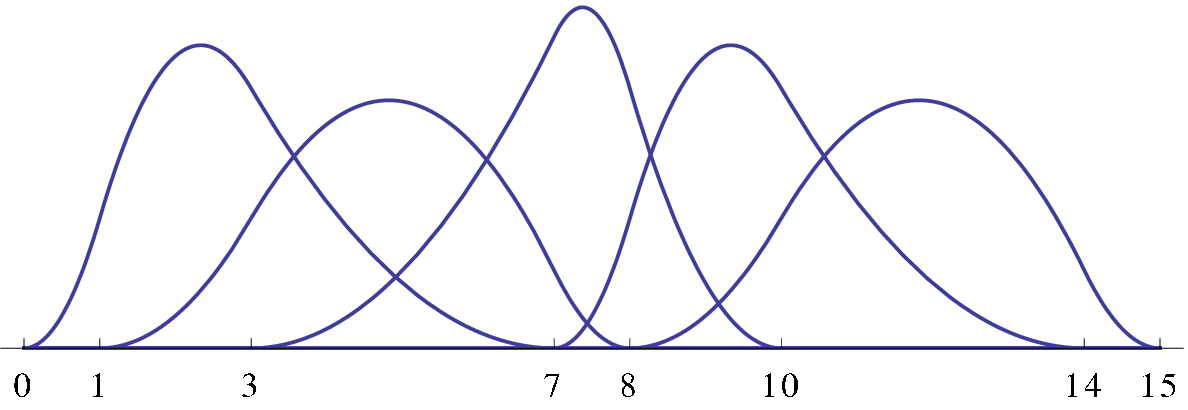}{8cm}
Each $M_i$ is strictly positive
on the interval $(x_i,x_{i+k})$ and zero outsize of it,
and on each interval $[x_j,x_{j+1}]$, each $M_i$ equals some
polynomial $p_{ij}$ of degree at most~$k-1$.

The characterization of $k$-monotone interpolability we will
use was obtained from Theorem~\ref{t:repr}
by a duality argument, and it can be stated as follows.

\begin{lemma}[\cite{kopotun-shadrin}, Corollary~6.5]
\label{l:KSchar}
Let $X=\{x_1,\ldots,x_{n+k}\}$, $x_1<\cdots<x_{n+k}$,
 be a node sequence, let $f\:X\to \R$ be a function, and let the vector
$v=(v_1,\ldots,v_n)$ be given by $v_i=[x_i,\ldots,x_{i+k}]f$.
Then $(X,f)$ is $k$-monotone interpolable if and only if
the following implication holds 
for every $a=(a_1,\ldots,a_n)\in\R^n$:
If $\sum_{i=1}^n a_i M_i(t)\ge 0$ for all $t\in [x_1,x_{n+k}]$,
then $\sum_{i=1}^n a_iv_i\ge 0$.
\end{lemma}

Geometrically, if we denote by $\MM$  the compact set
\[
\MM=\Bigl\{(M_1(t),\ldots,M_n(t))\in\R^n: t\in [x_1,x_{n+k}]\Bigr\},
\]
then the characterization says that $P$ is not $k$-monotone interpolable
if and only if the point $v$ can be strictly separated from $\MM$
by a hyperplane passing through the origin.

\section{Proof of Theorem~\ref{t:rams3} (Ramsey-type result)}\label{s:rams3}

The following alternative criterion for $k$-monotone interpolability
can be derived from the representation theorem (Theorem~\ref{t:repr})
or from Lemma~\ref{l:KSchar}.

\begin{lemma}\label{l:sufficient-condition}
 Let $X=\{x_1,\ldots,x_{n+k}\}$, $x_1<\cdots<x_{n+k}$,
 be a node sequence, let $f\:X\to \R$ be a function, and let the vector
 $v=(v_1,\ldots,v_n)$ be given by $v_i=[x_i,\ldots,x_{i+k}]f$.
 Then $(X,f)$ is $k$-monotone interpolable if and only if there exist $c_1,\ldots, c_n\geq 0$ and $t_1,\ldots, t_n\in[x_1,x_{n+k}]$
 satisfying $v_i=\sum_{j=1}^n c_j M_i(t_j)$ for all $i=1,\ldots, n$.
\end{lemma}

\begin{proof} The ``if'' part is obvious from Lemma~\ref{l:KSchar}:
the condition 
guarantees that $v$ lies in the convex cone generated by
the set $\MM$ defined after Lemma~\ref{l:KSchar}, and hence
it cannot be separated from~$\MM$. 

The ``only if'' part follows from a suitable hyperplane separation
theorem for convex cones; one needs to verify that the cone
generated by $\MM$ is closed. We omit the details since we
do not need the ``only if'' direction.
\end{proof}

\jirka{The previous proof of the lemma was a bit problematic, since
we did not check that the cone was closed (otherwise, we cannot
claim membership OR strict separation).}


We are now ready to prove the Ramsey-type result.

\begin{proof}[Proof Theorem \ref{t:rams3}]
Let $P = \{(x,f(x))\colon x\in X\}\subset \R^2$ be an $N$-point set with 
distinct $x$-coordinates. 

A necessary condition for $3$-monotone interpolability of $P$
is that $[x_0,\ldots,x_3]f\ge 0$ for every
choice of $x_0<\cdots<x_3\in X$.  This condition can be easily
enforced using Ramsey's theorem for fourtuples: we color
a fourtuple $\{x_0,\ldots,x_3\}\subseteq X$ red
if $[x_0,\ldots,x_3]f\ge 0$ and blue otherwise,
and if $N$ is sufficiently large, we can select a subset $Y\subseteq X$
of prescribed size in which all fourtuples have the same color.
By possibly passing to $P^\updownarrow$, we may thus assume
that $[x_0,\ldots,x_3]f\ge 0$ for all fourtuples in~$Y$.

Next, by Ramsey's theorem again, 
we will select an $(n+3)$-point subset $Z=\{z_1,\ldots,z_{n+3}\}\subseteq Y$, 
which we will prove to be $3$-monotone interpolable. This time we
will 2-color 5-tuples, in a way which looks mysterious at first sight,
but which will be explained by the proof below.

 For a node sequence $U=\{u_1 < \ldots < u_{m+3}\}$ of real numbers,
 let $M^U_i$ be $i$th B-spline of degree 2, i.e.,
 $[u_{i},u_{i+1},u_{i+2},u_{i+3}]\max(0,x-t)^{2}$. 
For $U\subseteq Y$, we also write 
$v_i^U$ for $[u_{i},u_{i+1},u_{i+2},u_{i+3}]f$.
 Note that our choice of  $Y$ guarantees $v_i^U\geq 0$
for every $U\subseteq Y$ and all~$i$.

Now we define the $2$-coloring of the $5$-tuples:
 a $5$-tuple $U = \{u_1 < \cdots < u_5\}\subseteq Y$ is \emph{v-positive} 
if 
\[
\frac{v_1^U}{M^U_1(u_3)}\leq \frac{v_2^U}{M^U_2(u_3)},
\] and
 otherwise it is \emph{v-negative}. 

We recall that $M_i^U(u_{j})$ is strictly positive
for $j=i+1$ and $j=i+2$ and zero for all other $j$,
and so the coloring is well defined.

 By Ramsey's theorem, if $|Y|$ is sufficiently large, 
there exists $Z=\{z_1<\cdots< z_{n+3}\}\subseteq Y$ with all $5$-tuples 
of the same type (i.e. either all v-positive  or all v-negative).
We will use Lemma~\ref{l:sufficient-condition} with $X=Z$ to show
that $Z$ is $3$-monotone interpolable.
From now until the end of the proof, to simplify
the notation, let us write $M_i$ for $M_i^Z$ and $v_i$ for $v_i^Z$.

 \paragraph{The v-positive case.}  Here we choose $t_j:=z_{j+2}$,
$j=1,\ldots,n$,  in Lemma~\ref{l:sufficient-condition}. With the
$t_j$ fixed, the conditions $v_i=\sum_{j=1}^n c_j M_i(t_j)$
provide a system of $n$ linear equations for the unknowns $c_1,\ldots,c_n$.

The idea is to calculate $c_1$, then $c_2$, then $c_3$, etc. from these
linear equations. In the $i$ step, $i\ge 2$, v-positivity is
 exactly the right condition for ensuring that $c_i\ge 0$.

Since $M_i(z_{j+2})$ is zero unless $j\in \{i-1,i\}$, the first
equation reads $v_1=c_1 M_1(z_3)$ and determines 
$c_1=v_1/M_1(z_3)$ uniquely.
We also have $c_1\ge 0$ since $v_i\ge 0$, by the choice of~$Y$.

Now we suppose inductively
that nonnegative $c_1,\ldots,c_{i}$ have been determined, in such a way
that they satisfy the first $i$ equations. Moreover, to support
the induction, we also assume $c_i\leq {v_i}/{M_i(z_{i+2})}$.

Then expressing $c_{i+1}$ from the $(i+1)$st equation gives
 \[c_{i+1}:=\frac{v_{i+1}-c_iM_{i+1}(z_{i+2})}{M_{i+1}(z_{i+3})}.
\]
Since $c_i\ge 0$ and $M_{i+1}\ge 0$, this formula implies
the inequality $c_{i+1}\leq {v_{i+1}}/M_{i+1}(z_{i+3})$
needed for our induction.

It remains to verify that $c_{i+1}\ge 0$, and here 
we use the v-positivity of the $5$-tuple  
$\{z_i,z_{i+1},z_{i+2},z_{i+3},z_{i+4}\}$, from which we obtain
\[
c_i\le \frac{v_i}{M_i(z_{i+2})} \le \frac{v_{i+1}}{M_{i+1}(z_{i+2})}.
\]
Hence the numerator in the formula for $c_{i+1}$ is nonnegative.
This finishes the inductive step; we have shown that
the condition in Lemma~\ref{l:sufficient-condition} is fulfilled
and so the restriction of $f$ to $Z$ is $3$-monotone interpolable.

 \paragraph{The v-negative case. }
 This case is similar to the previous one, 
but this time we set $t_j:=z_{j+1}$ (as opposed to $t_j=z_{j+2}$ in
the previous case), and we work backwards, computing first $c_n$,
then $c_{n-1}$, etc.

From the $n$th equation  we obtain $c_n={v_n}/{M_n(z_{n+1})}$.
In the inductive step, we assume that nonnegative $c_n,\ldots,c_{i+1}$
have been determined satisfying the last $n-i$ equations
and such that $c_{i+1}\le v_{i+1}/M_{i+1}(z_{i+2})$.
Then the $i$th equation dictates that
\[
c_i:=\frac{v_i-c_{i+1}M_i(z_{i+2})}{M_i(z_{i+1})}.
\]
As before, $c_i\le v_i/M_i(z_{i+1})$ follows immediately.
The v-negativity of $\{z_i,\ldots,z_{i+4}\}$ then yields
\[
c_{i+1}\le \frac{v_{i+1}}{M_{i+1}(z_{i+2})}<
\frac{v_{i}}{M_{i}(z_{i+2})},
\] 
again showing the numerator in the formula for $c_i$ nonnegative.
This concludes the proof.
\end{proof}

\section{Constructions of point sets}

We are going to prove Theorems~\ref{t:nohelly} and~\ref{t:expdigits}.
The idea of both constructions is similar, and first we prepare a
result common for both of them. But while it is possible to 
arrange the construction for Theorem~\ref{t:expdigits} so that
it also verifies Theorem~\ref{t:nohelly}, the technical details
come out complicated, and so we prefer to keep the two constructions
separate.

For a point $p \in \R$, we write $x(p)$ and $y(p)$ for the $x$ and $y$ 
coordinates of $p$.
\begin{lemma}
\label{l:constr-base}
Let $P = \{p_1, \ldots, p_n\}$ be a $3$-monotone interpolable point set 
where  $x(p_1) < \cdots < x(p_n)$.
Assume that for some parabola $\pi$, there is a $3$-monotone interpolant $f$ 
of $P$ equal to $\pi$ to the right of $p_n$.
Also assume that for a point $q$ to the right of $p_n$, 
$P \cup \{q\}$ is $3$-monotone interpolable if and only if $q$ lies on 
or above $\pi$.
Further, let $Q = \{q_1, q_2\}$ be a pair of points above $\pi$ that satisfy 
$x(p_n) < x(q_1) < x(q_2)$ and such that there is a parabola $\rho$ passing 
through $q_1$ and $q_2$ tangent to $\pi$ with the point of tangency to 
the right of $p_n$ and to the left of $q_1$ (see Fig.~\ref{f:parab-touch}
left).

Then for a point $r$ to the right of $q_2$, $P \cup Q \cup \{r\}$ is 
$3$-monotone interpolable if and only if $r$ lies on $\rho$
or above it.
\end{lemma}

\labfig{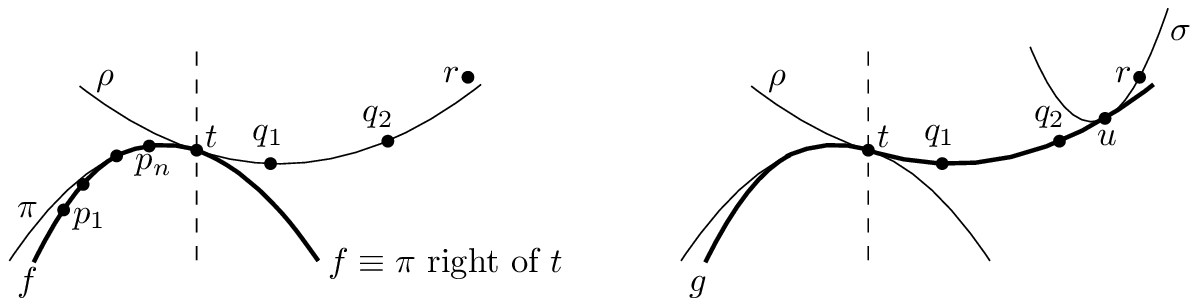}{Proof of Lemma~\ref{l:constr-base}.}

\begin{proof}
Let $t$ be the point of tangency of $\pi$ and $\rho$.
Notice that the curve $g$ equal to $f$ to the left of $t$ and equal to 
$\rho$ to the right of $t$ is a $3$-monotone interpolant of $P \cup Q$;
see Fig.~\ref{f:parab-touch}.
Consequently, if $r$ lies on $\rho$ to the right of $q_2$, 
then $P \cup Q \cup \{r\}$ is $3$-monotone interpolable.

Now assume that $r$ lies to the right of $q_2$ and above $\rho$. 
For every point $u$ on $\rho$ with $x(u) \neq x(r)$, 
there is a (unique) parabola passing through $r$ that is tangent 
to $\rho$ with $u$ as the point of tangency.
We fix a point $u$ on $\rho$ with $x(q_2) < x(u) < x(r)$ and 
a parabola $\sigma$ passing through $r$ and tangent to $\rho$ in $u$.
The curve equal to $g$ to the left of $u$ and equal to $\sigma$ 
to the right of $u$ is a $3$-monotone interpolant of 
$P \cup Q \cup \{r\}$.

Now we consider $r'$ to the right of $q_2$ and below $\rho$.
We assume, for contradiction, that $P \cup Q\cup \{r'\}$ 
has a $3$-monotone  interpolant $h$. 
Let $\sigma'$ be the parabola containing $q_1$, $q_2$ and $r'$:
\immfig{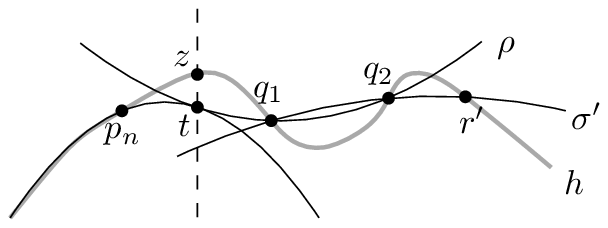}
Then $\sigma'$ and $\rho$ have exactly two points in common: 
$q_1$ and $q_2$.
Therefore $\sigma'$ is strictly below $\rho$ everywhere to the 
left of $q_1$.
We consider the point $z = (x(t), h(x(t)))$.
For the quadruple $\{z, q_1, q_2, r'\}$ to be positive, 
$z$ has to lie on $\sigma'$ or below it.
On the other hand, since $h$ is a $3$-monotone interpolant 
of $P \cup \{z\}$, $z$ lies on or above $\pi$.
This is a contradiction, since $\pi(x(t)) = \rho(x(t)) > \sigma'(x(t))$.
\end{proof}

Let $f,g:\R \rightarrow \R$ be two functions.
A \emph{convex combination} of $f$ and $g$ is the function 
$\alpha f + (1-\alpha) g$ for some $\alpha \in [0,1]$.

\begin{observation}
\label{o:conv-comb}
Let $k \ge 1$.
Let $f$ and $g$ be two $k$-monotone interpolants of a set $P$.
Then every convex combination of $f$ and $g$ is a $k$-monotone 
interpolant of $P$.
\end{observation}

\subsection{Proof of Theorem~\ref{t:nohelly} (non-locality)}\label{s:nonhelly}


We will prove the following by induction on $i$:

\begin{claim}
For every $i \ge 1$ there exists a set $P_i$  of $2i+1$ points in the plane
and an integer $u_i$ that satisfy the following.
There are quadratic functions $\pi_i$ and $\pi_{i,p}$, $p\in P_i$, 
where each $\pi_{i,p}(x) \le  \pi_i(x) - 1$ on $[u_i,\infty)$ for every $p \in P_i$,
such that:
\begin{enumerate}
\item[\rm(i)] There exists a $3$-monotone interpolant $f$ for $P_i$
that equals $\pi_i$ on $[u_i,\infty)$, but if $q$ is a point
with $x(q) \ge u_i$ and strictly below $\pi_i$, then
$P_i\cup\{q\}$ is not $3$-monotone interpolable.

\item[\rm(ii)]  
For every $p\in P_i$, the set $P_i\setminus\{p\}$ is
$3$-monotone interpolable, and among the $3$-monotone interpolants,
there is a function $f_{i,p}$ that equals $\pi_{i,p}$ on $[u_i,\infty)$.
\end{enumerate}
Moreover, the coordinates of all the points in $P_i$ are integers from 
the range $0,  \ldots, 25i^3$ and $\pi_i(u_i)$ is an integer.
\end{claim}

\begin{proof}
Define $u_i = 5i$.

When $i=1$, the requirements are satisfied by the triple of points $(0,0)$,
$(1,0)$, $(2,0)$.

For $i \ge 2$, we proceed by induction.
We have a set $P_{i-1}$ of $2i-1$ points and quadratic functions $\pi_{i-1}$
and $\pi_{i-1,p}$ for every $p \in P_{i-1}$.

We define 
\[
\pi_i(x) = \pi_{i-1}(x) + (x-u_{i-1}-1)^2.
\]
Thus, $\pi_i$ is a parabola tangent to $\pi_{i-1}$ at a point with 
$x$-coordinate $u_{i-1}+1$.
We also have $\pi_i(x) > \pi_{i-1}(x)$
for every $x \in \R \setminus \{u_{i-1}+1\}$.
We now define the set $P_i$ as $P_i = P_{i-1} \cup \{p_{2i}, p_{2i+1}\}$,
where $p_{2i}$ and  $p_{2i+1}$ are points on $\pi_i$ with $x$-coordinates 
$u_{i-1}+2$ and $u_{i-1}+3$.

Claim \rm(i) follows from Lemma~\ref{l:constr-base}.

Now we verify claim \rm{(ii)}. 

If $p \in P_{i-1}$, we consider the $3$-monotone interpolant 
$f_{i-1,p}$ of $P \setminus \{p\}$ that equals $\pi_{i-1, p}$ on 
$[u_{i-1}, \infty)$.
We define the parabola $\pi_{i,p}$ as the parabola that passes through 
$(u_{i-1}, \pi_{i-1}(u_{i-1}))$, $p_{2i}$ and $p_{2i+1}$.
That is, $\pi_{i} (x) - \pi_{i,p}(x)$ is a quadratic function that attains
its minimum at $u_{i-1}+5/2$ and is equal to $1$ at $u_{i-1}$.
Then we have $\pi_{i,p}(x) \le  \pi_i(x) - 1$ on $[u_i,\infty)$.
We also deduce
\[
\pi_{i,p}(x) = \pi_{i-1}(x) + \frac56 (x-u_{i-1})^2 - \frac76 (x-u_{i-1}).
\]
So we have $\pi_{i,p} (x) > \pi_{i-1}(x)-1 \ge \pi_{i-1,p}(x)$ for 
every $x \in [u_{i-1}, \infty)$.

Since $\pi_{i,p}$ has two intersections with $f_{i-1}$ and no intersection 
with $f_{i-1, p}$ on $[u_i, \infty)$,
$f_{i-1}$ and $f_{i-1, p}$ have a convex combination $g$ 
whose restriction on $[u_{i-1}, \infty)$ is a parabola tangent to $\pi_{i,p}$.
By Observation~\ref{o:conv-comb}, $g$ is a $3$-monotone interpolant of 
$P \setminus \{p\}$.

Let $t'$ be the point of tangency.
The function $f_{i,p}$ equal to $g$ on $(-\infty, x(t')]$ and equal to 
$\pi_{i,p}$ on $[x(t'), \infty)$ is a $3$-monotone interpolant for 
$P_i \setminus \{p\}$ that equals $\pi_{i,p}$ on $[u_i, \infty)$.

If $p = p_{2i}$ or $p = p_{2i+1}$, we let $q$ be the point from 
$\{p_{2i}, p_{2i+1}\}$ different from $p$.
We take the parabola $\pi_{i,p}$ that passes through $q$ and is tangent to 
$\pi_{i-1}$ at a point $t'$ with $x(t') = u_{i-1}$.
We have
\begin{align*}
\pi_{i,p}(x) &= \pi_{i-1}(x) + \frac49 (x-u_{i-1})^2 \text{ when } p = p_{2i} \text{ and} \\
\pi_{i,p}(x) &= \pi_{i-1}(x) + \frac14 (x-u_{i-1})^2 \text{ when } p = p_{2i+1}.
\end{align*}
In both cases, $\pi_{i,p}(x) \le  \pi_i(x) - 1$ on $[u_i,\infty)$.
The function $f_{i,p}$ equal to $f_{i-1}$ on $(-\infty, u_{i-1}]$ and equal to 
$\pi_{i,p}$ on $[u_{i-1}, \infty)$ is a $3$-monotone interpolant for 
$P_i \setminus \{p\}$ that satisfies claim \rm{(ii)}.

The $x$-coordinates of all the points of $P_i$ are integers from $\{0, \ldots, 5i\}$
and lie on the parabolas $\pi_i$.
All coefficients of the quadratic functions $\pi_i$ are integers and so
all the points in $P_i$ have integer coordinates.
We have $\pi_1 \equiv 0$ and
for every integer $i$ and every real $x \in [0, 5i]$, we have 
$\pi_i(x) \le  \pi_{i-1}(x) + (5i)^2$ and so $\pi_i(x) \le 25i^3$.
\end{proof}

We are now ready to prove Theorem~\ref{t:nohelly}.
We have $n=2i+2$ for some $i\ge 1$.
The set $P$ is formed by all the points of $P_i$ and a point 
$q = (u_i, \pi_{i}(u_i)-1)$.

\subsection{Proof of Theorem~\ref{t:expdigits} (doubly exponentially 
small example)}

\begin{lemma}
\label{l:par-eps-close}
Let $\eps > 0$.
For every $j$, let $p_j$ be the point $(j,j^3)$
and let $q_j$ be the point $(j,j^3+\eps)$.
Given an arbitrary integer $i$, let $\kappa$ be the parabola passing through
$q_{i-2}$, $p_{i-1}$ and $p_{i}$.
Then $(i+1)^3 - \kappa(i+1) = 6-\eps$ and $(i+2)^3 - \kappa(i+2) = 24-3\eps$.
\end{lemma}

\begin{proof}
We first consider the parabola $\tau_i$ passing through $p_i$, $p_{i+1}$ and
$p_{i+2}$ and a parabola $\tau_{i-2}$ passing through $p_{i-2}$, $p_{i-1}$ and
$p_{i}$.
By a straightforward calculation, for every $x \in \R$,
\[
\tau_i(x) = (3i+3)x^2 - (3i^2+6i+2)x + i^3+3i^2+2i
\]
and
\[
\tau_{i}(x) - \tau_{i-2}(x) = 6(x-i)^2.
\]
Let $\delta$ be the quadratic function $\kappa - \tau_{i-2}$. We have
$\delta(i-2) = \eps$, $\delta(i-1) = 0$ and $\delta(i) = 0$.
Thus, for every $x \in \R$:
\[
\kappa(x) - \tau_{i-2}(x) = \delta(x) = \frac{\eps}{2} \cdot (x-i+1/2)^2 - \eps/8.
\]

It is now easy to calculate the values $\tau_{i}(x) - \kappa(x)$ for $x = i+1$
and $x = i+2$ and verify the claim.
\end{proof}

\begin{lemma}
\label{l:par-tangent}
For every $j$, let $p_j$ be the point $(j,j^3)$.
For an integer $i>2$ and an arbitrary $\eps \in (0,1]$, let $q_{i+1}$ be the
point $(i+1, (i+1)^3 - 6 + \eps)$.
Let $\tau$ be the parabola passing through $p_{i-1}$, $p_{i}$ and $q_{i+1}$.
Then there is a parabola $\pi$ passing through $p_{i+1}$ and $p_{i+2}$ that is
tangent to $\tau$ such that the $x$-coordinate of the point of tangency is in
the interval $(i,i+1)$.
Moreover, $\pi(i+3) = (i+3)^3 - 6 + \delta$, where 
$\delta \in (0, \eps^2 /5)$.
\end{lemma}

\begin{proof}
Let $p'_{i+1} = (1, (i+1)^3-\tau(i+1))$ and 
$p'_{i+2} = (2,(i+2)^3-\tau(i+2))$.
From Lemma~\ref{l:par-eps-close}, we have $p'_{i+1} = (1, 6-\eps)$ and
$p'_{i+2} = (1, 24 - 3\eps)$.
The main part of the proof is finding a parabola $\pi'$ passing through
$p'_{i+1}$ and $p'_{i+2}$ that is tangent to the $x$-axis in a point with
$x$-coordinate in $(0,1)$.
Then we show that the parabola $\pi$ defined by $\pi(x) = \pi'(x-i) + \tau(x)$
for every $x \in \R$ has the claimed properties.

Since $p'_{i+2}$ is higher than $p'_{i+1}$ and both are above the $x$-axis,
there are exactly two parabolas passing through $p'_{i+1}$ and $p'_{i+2}$ that
are tangent to the $x$-axis. 
The point of tangency of one of the two parabolas is between $p'_{i+1}$ and
$p'_{i+2}$, while the point of tangency of the other is to the left of
$p'_{i+1}$.
The parabola with tangency to the left of $p'_{i+1}$ goes below the other
parabola everywhere to the left of $p'_{i+1}$ and thus has a smaller coefficient
of the quadratic term.

We write $\pi'(x) = ax^2 + bx + c$. Since $\pi'$ passes through $p'_{i+1}$ and
$p'_{i+2}$ and is tangent to the $x$-axis, we have
\begin{align*}
a+b+c &= 6-\eps \\
4a+2b+c &= 24-3\eps \\
b^2 &= 4ac.
\end{align*}

To simplify the equations, we define $\bar{a} = a - 6$.
Using the first two equations, we express $b$ and $c$ in terms of $\bar{a}$ and
$\eps$ as $b = -3\bar{a} -2 \eps$ and $c = 2\bar{a} + \eps$.
The third equation then becomes
\[
\bar{a}^2 + 8 \bar{a} \eps - 48 \bar{a} + 4 \eps^2 - 24 \eps = 0.
\]

Let $f(\bar{a})$ be the left-hand side of the equation.
Using $\eps \in (0,1]$, it is easy to verify that $f(-\eps/2) > 0$,
$f(0) < 0$ and that $f$ goes to infinity as $\bar{a}$ goes to infinity.
Let $\bar{a}_1$ and $\bar{a}_2$ be the two roots of $f(\bar{a})$ with
$\bar{a}_1 < \bar{a}_2$.
Since the value of $\bar{a}$ corresponding to the parabola $\pi'$ is the
smaller of the two roots of $f(\bar{a})$, its value is $\bar{a}_1 \in
(-\eps/2, 0)$.
We then have $a \in (5, 6)$ and $b \in (-2 \eps, -\eps/2)$.

The $x$-coordinate of the point of tangency of $\pi'$ with the $x$-axis is
\[
\frac{-b}{2a} \in \left(0, \frac{\eps}{5}\right) \subset (0,1).
\]
From $b^2 = 4ac$, we obtain
\[
c = \frac{b^2}{4a} \in \left(0, \frac{\eps^2}{5}\right).
\]
We define $\delta = c$. 
Notice that $\pi'$ passes through the point $(0,\delta)$.

Consequently, the parabola $\pi$ passes through $p_{i+1}$ and $p_{i+2}$ and is
tangent to $\tau$ in a point with $x$-coordinate in the interval $(i, i+1)$ and
passes through the point $(i,i^3+\delta)$.
By Lemma~\ref{l:par-eps-close}, $\pi(i+3) = (i+3)^3 - 6 + \delta$.
\end{proof}

The next lemma is a slight strengthening of Theorem~\ref{t:expdigits}.

\begin{lemma}
\label{l:expodigits}
Let $p_j$ be the point $(j,j^3)$ and let $z = (-1,0)$.
Let $P_m = \{z, p_0, p_1, \ldots, p_{2m+1}\}$.
For every integer $m \ge 0$, we consider the point $q_{2m+2}$ with
$x$-coordinate $2m+2$ and with the smallest possible $y$-coordinate such that
the set $P_m \cup \{q_{2m+2}\}$ is $3$-monotone interpolable.
Then the $y$-coordinate of $q_{2m+2}$ equals $(2m+2)^3-6+\eps_m$ for some
positive $\eps_m \le 2 \cdot 2^{-2^{m}}$.
\end{lemma}

\begin{proof}
Let $\pi_0$ be the parabola passing through $z$, $p_0$ and $p_1$.  
Observe that $\pi_0(2)=3$ and thus the claim holds for $m=0$ with $\eps_0=1$.

We now consider the inductive step for $m \ge 1$.

Let $\pi_{m-1}$ be the parabola passing through $p_{2m-2}$, $p_{2m-1}$ and
$q_{2m}$. 
As a consequence of the induction hypothesis, for every point $s$ to the right
of $p_{2m-1}$, $P_{m-1} \cup \{s\}$ is $3$-monotone interpolable if and only
if $s$ lies on or above $\pi_{m-1}$.

By Lemma~\ref{l:par-tangent}, there is a parabola $\pi_m$ passing through
$p_{2m}$ and $p_{2m+1}$ that is tangent to $\pi_{m-1}$ in a point to the left
of $p_{2m}$ and to the right of $p_{2m-1}$.

By Lemma~\ref{l:constr-base}, for every point $s$ to the right of $p_{2m+1}$, 
$P_m \cup \{s\}$ is $3$-monotone interpolable if and only if $s$ lies 
above $\pi_m$.
By Lemma~\ref{l:par-tangent}, $\pi_m(2m+2) = (2m+2)^3-6+\eps_m$ for some 
$\eps_m \in (0, \eps_{m-1}^2/5)$. That is,
\[
\eps_m \le \frac{\eps_{m-1}^2}{5} \le \frac{(2 \cdot 2^{-2^{m-1}})^2}{5} \le
\frac{4}{5} \cdot 2^{2 \cdot (-2^{m-1})} \le 2 \cdot 2^{-2^m}.
\]
\end{proof}


\section{Proof of Theorem~\ref{t:expdigits} (exponentially many digits)}

\subsection{The semidefinite formulation}\label{s:sdp}

By the characterization in Lemma~\ref{l:KSchar}, if we think
of a point set $P\subset\R^2$ as a function $f\:X\to\R$,
with $X=\{x_1,\ldots,x_{n+k}\}$, then $(X,f)$ is \emph{not}
$k$-monotone interpolable exactly if there is $a\in\R^n$
such that $\sum_{i=1}^n a_i M_i(t)\ge 0$ for all $t\in [x_1,x_{n+k}]$
and $\sum_{i=1}^n a_iv_i=-1$, where the $v_i=[x_i,\ldots,x_{i+k}]f$
are the $k$th divided differences. 
Further we recall that $M_i(t)$ equals a polynomial 
$p_{ij}(t)$ of degree at most $k-1$ on each interval $[x_j,x_{j+1}]$.

By re-scaling the interval $[x_j,x_{j+1}]$ to $[-1,1]$ for notational
convenience, each $p_{ij}$ is transformed into another polynomial
$\tilde p_{ij}$. The coefficients of $\tilde p_{ij}$ can obviously
be computed from the $x_i$ in polynomial time. Thus, the impossibility
of $k$-monotone interpolation is a special case of the following 
computational problem.

\begin{problem}[The non-positivity\footnote{The word positivity
refers to a customary terminology: a vector $v$ is called
\emph{positive} w.r.t. a system $u_1,\ldots,u_n$ of real
functions on an interval $I$ if $\sum_{i=1}^n a_iu_i(t)\ge 0$
for all $t\in I$ implies $\sum_{i=1}^n a_iv_i\ge 0$.}
problem]\label{p:pos}\ 

\noindent\emph{Input:} 
Polynomials $\tilde p_{ij}(t)$, $i=1,2,\ldots,n$, $j=1,2,\ldots,m$
with rational coefficients and a vector $v\in\Q^n$.

\noindent\emph{Question:} Does there exist $a\in\R^n$ such that
$\sum_{i=1}^n a_i\tilde p_{ij}(t)\ge 0$ for all $t\in[-1,1]$ and
all $j=1,\ldots,m$,
and $\sum_{i=1}^n a_iv_i=-1$?
\end{problem}

There is a large body of work showing that problems involving 
nonnegativity of polynomials over semialgebraic sets (i.e., sets
defined by polynomial inequalities) can be converted, under fairly
general conditions, to semidefinite programs.

\heading{Semidefinite programs. } 
We recall that a \emph{semidefinite program} 
is the computational problem of finding a positive semidefinite
$n\times n$ matrix $X$ that maximizes a linear function 
$C\bullet X$ subject to linear
constraints $A_1\bullet X=b_1$,\ldots, $A_2\bullet X=b_m$, for given
$n\times n$ matrices $C$ and $A_1,\ldots,A_m$ and reals $b_1,\ldots,b_m$.
Here the matrix scalar product $\bullet$ is defined
as $C\bullet X=\sum_{i,j=1}^n c_{ij}x_{ij}$.
We refer, e.g., to the books
\cite{NemiroBenTal,BoydVand,GaertMat} or handbooks 
\cite{sdp-handbook,handbo-optim} for background.

For the semidefinite formulation of our non-positivity problem,
the maximized function $C\bullet X$ is irrelevant;
we need only the \emph{semidefinite feasibility problem},
where we ask for the existence of a positive semidefinite $X$
satisfying given linear constraints.

\heading{Semidefinite formulation of the non-positivity problem. }
By a classical result, see
 \cite[Theorem~2.6]{lasserre-book}, a univariate polynomial $p(t)$
of degree $d$ is nonnegative on $[-1,1]$ if and only if it can be written as
\[
p(t)=f(t)+(1-t)(1+t)h(t), 
\]
where $f(t)$ and $h(t)$ are polynomials 
that can be expressed as sums of squares of 
suitable polynomials, i.e., in the form $\sum_{i=1}^m s_i(t)^2$
for some $m$ and some polynomials $s_1(t),\ldots,s_m(t)$,
with $\deg f\le 2d$ and $\deg h\le 2d-2$.

Moreover, a polynomial $f(t)$ is a sum of squares of degree at most $2d$ iff
it has the form $\ttt^TQ \ttt$, where $Q$ is a $(d+1)\times(d+1)$
positive semidefinite matrix and $\ttt=(1,t,t^2,\ldots,t^d)$;
see \cite[Prop.~2.1]{lasserre-book}. 

Thus, a  polynomial $p(t)$ of degree at most $d$
is nonnegative on $[-1,1]$ if and only if there are
a $(d+1)\times (d+1)$ matrix $Q$ and $d\times d$ matrix
$\tilde Q$, both positive semidefinite, such that
\[
p(t)=\ttt^TQ \ttt+(1-t)(1+t) \tilde\ttt^T\tilde Q \tilde\ttt
\]
holds as equality of polynomials in $t$, where 
$\tilde\ttt=(1,t,\ldots,t^{d-1})$. Expanding each side according to powers
of $t$, we obtain $2d+1$ linear equations involving the entries of
$Q$ and $\tilde Q$ and the coefficients of~$p(t)$.

Therefore, the non-positivity problem above can be re-stated as
the existence of reals $a_1,\ldots,a_n$ and positive semidefinite
matrices $Q_1,\ldots,Q_m$ (of size $k\times k$)
and $\tilde Q_1,\ldots\tilde Q_m$ (of size $(k-1)\times(k-1)$)
such that $\sum_{i=1}^n a_iv_i=-1$ and for each $j=1,2,\ldots,m$,
the matrices $Q_j$ and $\tilde Q_j$ witness the nonnegativity
of the polynomial $\sum_{i=1}^n a_i\tilde p_{ij}(t)$ in the above sense,
using suitable linear equations involving the entries of $Q_j$ and
$\tilde Q_j$ and the~$a_i$.

This is not yet quite a semidefinite feasibility problem as defined above,
but it can be transformed into one by standard tricks. Namely,
we first replace each of the scalar variables $a_i$ by the
difference $a_i'-a_i''$, where $a'_i$ and $a''_i$ are new
nonnegative scalar variables. Then
we set up a large block-diagonal matrix $X$ that has the matrices
$Q_1,\ldots,Q_m$ and $\tilde Q_1,\ldots\tilde Q_m$ on the diagonal,
as well as the $1\times 1$ blocks containing $a'_1,a''_1,\ldots,
a'_m,a''_m$, and zeros elsewhere. The zeros are forced as linear
equalities, of the form $A_j\bullet X=0$, for the entries of $X$.
As is well known, positive semidefiniteness
of $X$ is equivalent to positive semidefiniteness of all the $Q_j$
and $\tilde Q_j$ plus the nonnegativity of the $a'_i$ and $a''_i$.
In this way, we get a semidefinite feasibility problem, whose input
size is bounded by a polynomial in $k$ and in the input size of
the non-positivity problem.

We will refer to the resulting semidefinite feasibility problem
as the \emph{standard semidefinite formulation} of the non-positivity
problem (or of the $k$-monotone interpolability problem).

\heading{Feasible solutions requiring exponentially many digits. } 
Theorem~\ref{t:expdigits},
the example of a non-interpolable set for which a set lying extremely
close is interpolable, yields the following consequence.

\begin{corol}\label{c:expdigiSDP} For the $3$-monotone noninterpolable point
set $P_m$, $m\ge 2$, as in Theorem~\ref{t:expdigits} (with $O(m)$ points
with integer coordinates bounded by $O(m^3)$), every vector $a$
as in the corresponding non-positivity problem (Problem~\ref{p:pos})
has entries exceeding $2^{2^m}/100m$ in absolute value. Consequently,
every feasible solution of the standard semidefinite formulation
has components with exponentially many digits.
\end{corol}

\begin{proof}  Let $a$ be a vector as in Problem~\ref{p:pos},
witnessing the $3$-monotone 
non-interpolability of $P_m$, and let $A=\|a\|_\infty=\max_i|a_i|$.

Let $P'_m$ be the $3$-monotone interpolable set as in 
Theorem~\ref{t:expdigits}. Let $v$ be the vector of the $k$th divided
differences for $P_m$ and $v'$ the one for $P'_m$. Since the $y$-coordinates
of $P_m$ and of $P'_m$ differ by at most $\eps:=2\cdot 2^{-2^m}$
and the $x$-coordinates are integers,
from the definition of divided differences
 it is easy to check that $\|v-v'\|_\infty\le 8\eps$.

Hence, with $n=|P_m|=2m+3$,
we have $\sum_{i=1}^n a_iv'_i\le \sum_{i=1}^n a_iv_i +nA\cdot8\eps$.
If we had $A\le (8n\eps)^{-1}$, then $\sum_{i=1}^n a_iv'_i<0$,
and so $a$ would also witness
non-interpolability of $P'_m$. The corollary follows.
\end{proof}

\heading{A simpler example for a variant of the non-positivity
problem. } If we take the non-positivity problem for general
quadratic polynomials $\tilde p_{ij}$, not necessarily coming
from $3$-monotone interpolability, there is a simpler example
forcing exponentially many digits. 

For simplicity, we replace the condition $t\in[-1,1]$ with $t\in\R$.
A quadratic polynomial $At^2+Bt+C$ is nonnegative on $\R$ if and only if
$A\ge 0$ and $B^2\le 4AC$.

Let us set $v=(-1,0,\ldots,0)$; then
$\sum_{i=1}^n a_iv_i$ forces $a_1=1$. Clearly, the polynomials 
$\tilde p_{ij}(t)$
can be set so that the polynomials $q_j(t):=\sum_{i=1}^n a_i \tilde p_{ij}(t)$
are as follows: $q_1(t)=a_2-2a_1$, and $q_i(t)=a_{i+1}t^2+2a_it+a_{1}$ for
 $i=2,3,\ldots,m$.
The nonnegativity of $q_1$ makes sure that $a_2\ge 2$, and nonnegativity of 
$q_i$ yields $4a_{i}^2\le 4a_{i+1}a_1$. Then we have
$a_i\ge 2^{2^{i-2}}$.

\subsection*{Acknowledgment}

We would like to thank Kirill Kopotun, Fadoua Balabdaoui, Jean B. Lasserre, 
and Mohb Safey El Din for kindly answering our questions, and 
Viola Mesz\'aros for useful discussions at the initial stages of this
research.

\bibliographystyle{alpha}
\bibliography{../cg,../geom}

\end{document}